\pdfoutput=1
\documentclass[12pt,letterpaper]{article}

\usepackage{cite}

\usepackage{hyperref}
\usepackage[small]{titlesec}

\usepackage{geometry} %
\usepackage{microtype}
\usepackage{mathtools}
\mathtoolsset{mathic}
\usepackage{xcolor,latexsym,amscd,amssymb,amsthm,graphicx,subfigure,cite,url}
\usepackage[shortlabels]{enumitem}

\DeclareMathOperator{\polylog}{polylog}

\usepackage{fullpage}

\def\A{\mathcal{A}}
\def\E{\mathcal{E}}
\def\F{\mathcal{F}}

\def\T{\mathcal{T}}

\def\OO{{\cal O}}
\def\VV{{\cal V}}

\let\bd\partial
\let\eps\varepsilon
\def\reals{\mathbb{R}}
\newcommand{\ints}{\mathbb{Z}}

\def\slice{V}

\newtheorem{theorem}{Theorem}[section]
\newtheorem{lemma}[theorem]{Lemma}

\newtheorem{fact}[theorem]{Fact}

\newtheorem{prop}[theorem]{Proposition}

\theoremstyle{remark}
\newtheorem*{remark}{Remark}

\begin{document}

\begin{titlepage}
\title{\Large Eliminating Depth Cycles\\among Triangles in Three Dimensions%
  \thanks{%
    Work on this paper by B.A.\ has been partially supported by NSF
    Grants CCF-11-17336, CCF-12-18791, and CCF-15-40656, and by BSF grant 2014/170.
    Work by M.S.\ has been supported by Grant 2012/229 from the
    U.S.-Israel Binational Science Foundation, by Grants 892/13 and 260/18 from
    the Israel Science Foundation, by the Israeli Centers for Research
    Excellence (I-CORE) program (center no.~4/11), by the Hermann
    Minkowski--MINERVA Center for Geometry at Tel Aviv University, and
    by Grant G-1367-407.6/2016 from the German-Israeli Foundation for Scientific Research and Development.
    An earlier version of this work appeared in \emph{SODA'17} \cite{triangle-cycles-soda}.}
}

\author{%
  Boris Aronov\thanks{%
    Department of Computer Science and Engineering,
    Tandon School of Engineering,
    New York University,
    Brooklyn, NY 11201, USA;
    \textsl{boris.aronov@nyu.edu}.}%
  \and
  Edward Y.\ Miller\thanks{%
    Courant Institute of Mathematical Sciences, 
    New York University,
    New York, NY 10012, USA;
    \textsl{emiller@cims.nyu.edu}.}%
  \and
  Micha Sharir\thanks{%
    Blavatnik School of Computer Science, Tel Aviv University, Tel-Aviv 69978,
    Israel; \textsl{michas@post.tau.ac.il}.} }

\maketitle

\begin{abstract}
  Given $n$ pairwise openly disjoint triangles in
  3-space, their vertical depth relation may contain
  cycles. We show that, for any $\eps>0$, the triangles can be cut
  into $O(n^{3/2+\eps})$ connected semi-algebraic pieces, whose description complexity depends only on the
  choice of $\eps$, such that the depth relation among these pieces is
  now a proper partial order. This bound is nearly tight in the worst
  case.  We are not aware of any previous study of this problem, in this full generality, with a
  subquadratic bound on the number of pieces.

  This work extends the recent study by two of the authors (Aronov, Sharir~2018) on
  eliminating depth cycles among lines in 3-space.  Our approach is
  again algebraic, and makes use of a recent variant of the polynomial
  partitioning technique, due to Guth, which leads to a recursive
  procedure for cutting the triangles.  In contrast to the case of
  lines, our analysis here is considerably more involved, due to the
  two-dimensional nature of the objects being cut, so additional
  tools, from topology and algebra, need to be brought to bear.

  Our result essentially settles a 35-year-old open problem in
  computational geometry, motivated by hidden-surface removal in
  computer graphics.
\end{abstract}
\end{titlepage}

\begin{flushright}
  To Ricky Pollack, a mathematician, an inspiration, a friend,\\who brought us all together.
\end{flushright}

\section{Introduction}
\label{sec:intro}

\paragraph*{The problem.}
Let $\T$ be a collection of $n$ non-vertical relatively open pairwise disjoint triangles in $\reals^3$. 
We~treat the triangles as relatively open, and allow their boundaries to intersect or overlap.
For any pair $\Delta,\Delta'$ of triangles in $\T$, we say that
$\Delta$ passes \emph{above} $\Delta'$ (equivalently, $\Delta'$ passes \emph{below}~$\Delta$) 
if there exists a vertical line that meets both $\Delta$ and $\Delta'$, so that it intersects $\Delta$ at a point that 
lies higher than its intersection with~$\Delta'$; this property is clearly independent of the 
choice of the vertical line meeting both triangles. We denote this relation by~$\Delta'\prec \Delta$ or~$\Delta\succ \Delta'$.
The~relation $\prec$ in general may contain \emph{cycles} of the form
$\Delta_1\prec\Delta_2\prec\cdots\prec\Delta_k\prec\Delta_1$. We call this a
\emph{$k$-cycle}, and refer to $k$ as the \emph{length} of the cycle.
Cycles of length three (the minimum possible length, for any collection of pairwise disjoint convex objects) are called \emph{triangular}; see Figure~\ref{fig:tricyc}.
\begin{figure}[hbpt]
  \centering
  \includegraphics{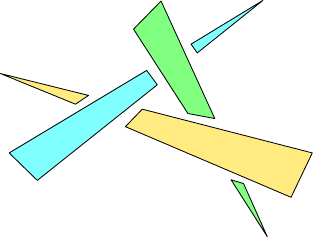}
  \caption{A triangular depth cycle among three triangles.}
  \label{fig:tricyc}
\end{figure}

The problem of \emph{cycle elimination} is to cut the triangles of $\T$ into a finite 
number of connected pieces, each being semi-algebraic of constant description
complexity (that is, defined by polynomial equalities 
and inequalities, with the number of inequalities and their maximum degree bounded by some constant), so that the suitably extended
depth relation among the new pieces is acyclic, in which case we call it a \emph{depth order}. 

The simpler case, with triangles replaced by lines or 
line segments, has been handled in the recent companion paper \cite{ArS1dcg}.
We refer the reader to that paper for a more detailed review of previous work.
An additional review of more recent progress and related results is given later in this paper.  
Note that eliminating cycles in a set of triangles adds, literally, a new dimension to the problem:
whereas lines (or segments, or curves) are cut at a discrete set of points, triangles have to be cut
into pieces along curves,
which makes the analysis considerably more involved.
We also observe that the binary space partition~(BSP) technique of Paterson and Yao \cite{PY} 
constructs a depth order by cutting the triangles into $\Theta(n^2)$ pieces,\footnote{%
  A significant feature of the BSP technique is that the cuts are made by straight lines and therefore the resulting pieces can be taken to be triangular, whereas this is not the case in our construction.  Moreover, the resulting pieces have no depth cycles with respect to \emph{any} viewing direction, and also for \emph{any perspective} view.}
but, as in the case of lines, we would like to use fewer cuts, ideally close to
the lower bound of $\Omega(n^{3/2})$, which is an immediate extension of a
similar lower bound in \cite{CEG+} for the case of lines.

A long-standing conjecture for the case of lines, open since 1980, has been that one can indeed always
construct a depth order by cutting the lines into a \emph{subquadratic} number of pieces.
See \cite[Chapter~9]{MdB-lncs} for a summary of the state of affairs circa 1990.
In the previous work \cite{ArS1dcg} we have shown that $O(n^{3/2}\polylog n)$ cuts (and thus pieces produced) suffice
to eliminate all cycles among $n$ lines in space. In this paper we obtain a similar,
albeit slightly weaker, bound for the case of triangles, settling this conjecture for the case of triangles, in a strong, almost worst-case tight manner.\footnote{%
  Except for the fact that our cuts are not by straight segments~--- see below.}

\paragraph*{Background.}
The main motivation for studying this problem comes from \emph{hidden
  surface removal} in computer graphics, as described, for example, in an  earlier paper of Aronov~et~al. \cite{AKS}.  Briefly, a conceptually simple technique for
rendering a scene in computer graphics is the so-called Painter's
Algorithm, which places the objects in the scene on the screen in a
back-to-front manner, painting each new object over the portions of
earlier objects that it hides.  For this, though, one needs an
\emph{acyclic} depth relation among the objects with respect to the
viewing point (which we assume hereafter, without loss of generality, to lie at
$z=+\infty$), as the algorithm would fail if applied directly to the triangles of Figure~\ref{fig:tricyc}, say. When there are cycles in the depth relation, one would
like to cut the objects into a small number of pieces, so as to
eliminate all cycles (i.e., have an acyclic depth relation among the
resulting pieces), and then paint the pieces in the above manner,
obtaining a correct rendering of the scene; see \cite{AKS,4M-book} for
more details.  Assuming that the input objects are all given as
triangulated polyhedral approximations, as is the case in many
practical applications, we face exactly the problem addressed in this
paper.

In the recent companion work \cite{ArS1dcg}, we essentially settled the case of lines, and thereby also of segments, by showing
that $O(n^{3/2}\polylog n)$ cuts are sufficient to eliminate all cycles, which is close to the best possible bound due to a well known construction requiring $\Omega(n^{3/2})$ cuts \cite{CEG+}. Refer to \cite{ArS1dcg} for more details on the history of the problem.

In contrast, the case of triangles has barely been touched, except for the work in \cite{AKS} just discussed, and for the aforementioned
BSP technique in \cite{PY} and several subsequent refinements, where improved 
(subquadratic) bounds were established for several special classes of objects in three dimensions, such as axis-parallel two-dimensional rectangles of bounded aspect ratio \cite{AGMV,T-BSP-fat} or so-called uncluttered scenes \cite{MdB-BSP-uncluttered}; see \cite{AT-BSP-survey,T-BSP-fat,T-BSP-MSR-survey} for surveys of the BSP literature.

\paragraph*{Our contribution.}
In this paper we essentially settle the problem for the case of triangles, and show that \emph{all} cycles
in the depth relation in a set of $n$ pairwise disjoint, relatively open triangles can be eliminated by cutting the triangles into $O(n^{3/2+\eps})$ 
connected pieces, for any choice of $\eps>0$, where the constant of proportionality depends on $\eps$,
and increases as $\eps$ tends to $0$.  The description complexity of the resulting pieces depends only on the choice of $\eps$.
As noted, our bound is best possible in the worst case, up to the $O(n^\eps)$ factor. 

The proof of this bound follows the high-level approach in the previous analysis
for the case of lines \cite{ArS1dcg}, which uses the polynomial partitioning technique
of Guth \cite{Gut}. Roughly speaking, this technique spreads the \emph{edges} of the 
triangles more or less evenly among the cells of the partition, which in turn provides a recursive 
divide-and-conquer mechanism for performing the cuts. However, the fact that we are dealing here with two-dimensional
triangles, rather than with one-dimensional lines (or segments), raises
substantial technical problems that need to be overcome.
The two most significant issues that arise are:

\begin{enumerate}[(i),wide,labelindent=0pt]
\item In contrast with the case of lines, where the cuts are made at a discrete set of points, 
here we need to cut the triangles into two-dimensional regions. Ideally, we would like to cut 
them into triangular pieces (as does the BSP technique of \cite{PY}), but our approach does not
achieve this and instead cuts the triangles by a collection of high-but-constant-degree algebraic 
curves into semialgebraic regions, which can then be cut into asymptotically the same number of 
`trapezoidal' regions of constant description complexity. 

\item Additional complications arise in controlling the recursive mechanism,
to ensure that not too many triangles are passed to a recursive subproblem (each within some cell of the partition).
As alluded to above, we can control the number of triangles that have an edge  crossing
a cell, since the partition is based on the triangle edges, but we do not have a good bound
on the number of triangles that ``fully slice'' through a cell; see the precise description below, and refer to Figure~\ref{fig:pierce} for an illustration.
\begin{figure}
  \centering
  \includegraphics[scale=1.2]{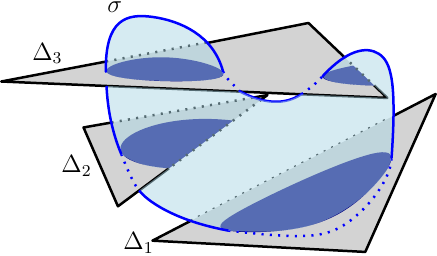}
  \caption{$\Delta_1$ slices the cell $\sigma$.  $\Delta_2$ pierces $\sigma$.  $\Delta_3$ pierces $\sigma$, even though only one connected component of $\Delta_3 \cap \sigma$ meets the boundary of $\Delta_3$.}
  \label{fig:pierce}
\end{figure}
One therefore needs to prune away the triangles that cross a cell in this slicing manner, 
in order to obtain a recurrence relationship similar
to the one in \cite{ArS1dcg} for the case of lines, thereby achieving the desired near-optimal
bound for the overall number of cuts.
\end{enumerate}

As in the case of lines, our proof is constructive, and leads, in principle, to an 
efficient algorithm for performing the cuts (assuming a suitable, by now standard, model of algebraic 
computation). The only ingredient that was missing --- an effective and efficient construction of 
Guth's partitioning polynomial --- has been resolved in \cite{AgArEzZ-poly}; see the discussion in Section~\ref{sec:discussion}.

\paragraph*{Very recent related work.}
We defer the discussion of extensions of our methods, related results, and research done after 
the initial version \cite{triangle-cycles-soda} of this work, to Section~\ref{sec:discussion}.
One significant later development of this kind is due to De Berg~\cite{mdb}, who presented a 
technique for eliminating all cycles in a set of $n$ triangles using only straight cuts, 
but the bound on the number of pieces produced by his technique is larger, about $O(n^{7/4})$.
Another step forward was the realization that Guth's partitioning theorem (see Proposition~\ref{prop:gut} below) can be made effective \cite{AgArEzZ-poly}.

\section{Eliminating cycles in a set of triangles}
\label{sec:tri}

\paragraph*{The setup and some notation.}
Let $\OO$ be a collection of pairwise disjoint \emph{objects} in three dimensions, 
where each object is a relatively open path-connected subset of a non-vertical plane;\footnote{%
 We assume, for simplicity of presentation, that no triangle lies in a vertical plane to avoid 
 some technicalities that do not essentially affect the rest of our argument.}
in our analysis, these will be the triangles or the 
triangle pieces produced by our construction.\footnote{%
  We choose to treat the objects as relatively open in order to make the vertical relation 
  unambiguous in situations where two objects touch, such as two triangles sharing an edge.
  In addition, this allows us to consistently treat portions into which a triangle is 
  cut as properly disjoint objects.}
Clearly, each object in~$\OO$ is 
\emph{$xy$-monotone}, that is, its intersection with any vertical line is a single point 
or empty.  Extending the definition given in the introduction, we define a \emph{depth relation} $(\OO,\prec)$ on the objects of~$\OO$, in the 
following natural manner: we say that $o_1 \in \OO$ \emph{lies} (or \emph{passes) below} 
$o_2 \in \OO$ (in which case we also say that \emph{$o_2$ passes above $o_1$}), and write 
$o_1 \prec o_2$ or $o_2 \succ o_1$, if there exists a~vertical line $\ell$ that meets both 
$o_1$ and $o_2$, and the $z$-coordinate of its intersection with $o_1$ is smaller than that 
of its intersection with $o_2$. For general planar connected regions, this relation 
need not be well behaved, but, for portions of pairwise disjoint triangles,
the relation is well defined, in the sense that it is independent of the choice of the line~$\ell$.
  
The final pieces into which the triangles of $\T$ will be cut will have 
\emph{constant description complexity}, as defined above.
However, until the very end of the construction, we will only generate certain constant-degree 
algebraic curves (or rather arcs)
that are drawn on the respective triangles.  Only upon termination of the process, we will use these
curves to construct the output collection of constant-complexity pieces with the desired properties. 

A \emph{cycle} in $(\OO,\prec)$ is a circular sequence of some $k$ objects from $\OO$ that satisfy 
$o_1\prec o_2 \prec \dots \prec o_k \prec o_1$. We refer to $k$ as the \emph{length} of the cycle; 
a cycle of length~$k$ is a \emph{$k$-cycle}. Note that self-loops and 2-cycles are not possible 
in $\OO$ under our assumptions (although they may very well exist for more general objects, already for algebraic arcs), 
so we must have $k\ge 3$. 

\paragraph*{The problem, restated.}
We are now ready to formally state the problem: 
Let $\T$ be a collection of $n$ non-vertical pairwise disjoint relatively open triangles in~$\reals^3$. 
As already mentioned above and illustrated in Figure~\ref{fig:tricyc}, $(\T,\prec)$ may 
contain cycles. Our goal is to cut the triangles of $\T$ into a small number of (relatively open) path-connected 
pieces of constant description complexity, so that, for the collection $\OO$ of the resulting
pieces, $(\OO,\prec)$ is acyclic---a \emph{depth order}.

A straightforward way of achieving this is to project all triangles of $\T$ orthogonally to the $xy$-plane and form the resulting arrangement of triangles, which consists of at most 
$O(n^2)$ faces.  Extrude each face of this arrangement into an unbounded $z$-vertical prism, cut each 
triangle $\Delta\in\T$ into pieces along the polygonal curve of its intersection with the prism boundary, and repeat this procedure for each prism.
It is easy to see that the number of resulting pieces is $O(n^3)$, a bound tight in the worst case for this specific construction, and that the 
pieces corresponding to a single prism form a linear order under $\prec$, while pieces from
different prisms are unrelated by $\prec$, so indeed there are no cycles.  It is moreover easy 
to refine this decomposition so that the resulting pieces are triangles, with no asymptotic 
increase in the number of pieces.

The cubic number of pieces obtained by this na\"ive approach is way too excessive.
A better bound on the number of pieces, sufficient to eliminate all cycles, is provided by the 
\emph{binary space partition}~(BSP) technique of Paterson and Yao \cite{PY}, which eliminates 
all cycles by cutting the triangles into $\Theta(n^2)$ (triangular) pieces. In fact, as already mentioned, the
construction in \cite{PY} has a much stronger property: the resulting collection of triangular
pieces has no cycles in the depth relation corresponding to \emph{any} viewing direction, or, 
more generally, to the perspective view from any point.

In this paper we show that cycles in the depth relation for a \emph{fixed} viewing direction 
(here, the view from $z=+\infty$) can be eliminated by creating a significantly subquadratic 
number of pieces, while keeping the complexity of each piece constant.
As already mentioned, the number of pieces that our technique yields, which is $O(n^{3/2+\eps})$, 
for any prespecified $\eps>0$, is nearly tight in the worst case.
The complexity of the pieces into which we cut our triangles depends \emph{only} on $\eps$.

We will cut the triangles by drawing curves on each of them;
this will be performed in a~hierarchical manner, by a recursive procedure.  
For each triangle $\Delta$, the curves drawn on $\Delta$ form a planar arrangement within
$\Delta$, and the overall collection of faces of these arrangements, over all $\Delta\in\T$,
will have an acyclic depth relation. 
In general, though, these faces need not have constant complexity, so a final step breaks them
into subfaces that do have constant complexity, without affecting the asymptotic bound on 
the number of pieces. The actual procedure that produces these curves is somewhat more
involved, and will be described in detail below.

Ideally, we would like the curves to be straight and the pieces to be triangular, 
as yielded by the BSP technique \cite{PY}. The very recent work of De~Berg \cite{mdb}, 
which is based on our earlier work~\cite{ArS1dcg}, also cuts the triangles by straight
segments, but it produces a larger number of pieces, with a bound close to $O(n^{7/4})$
(see Section~\ref{sec:discussion} for more details). Unfortunately, our argument 
cannot achieve this straightness property. 

So let $\OO$ denote the collection of faces of the arrangements induced by the 
curves drawn on the triangles of $\T$. Let $C$ be a cycle 
$o_1\prec o_2\prec\cdots\prec o_k\prec o_1$ in~$\OO$, with~$k\ge 3$.
We associate with $C$ a continuum $\Pi(C)=\Pi(C,\OO)$ of closed paths 
(to which we refer as \emph{loops}), where, informally, 
each loop $\pi$ in $\Pi(C)$ traces the cycle along the objects.
Formally, each such $\pi$ is defined in terms of $k$ vertical lines $\ell_1,\ldots,\ell_k$,
such that, for each $i$, $\ell_i$ intersects both $o_i$ and $o_{i+1}$ (where 
addition of indices is $\bmod\,k$), at respective points
$v_i^-$, $v_{i+1}^+$, so that $v_i^-$ lies below $v_{i+1}^+$. 
For each $i$, we connect the two points $v_i^+, v_i^- \in o_i$ 
by a Jordan arc $\pi_i \subset o_i$. The loop $\pi$ is then the cyclic concatenation
\begin{equation} \label{eq:pi}
\pi = \pi_1 \parallel v_1^-v_2^+ \parallel\pi_2 \parallel v_2^-v_3^+ \parallel 
\cdots \parallel v_{k-1}^-v_k^+ \parallel \pi_k \parallel v_k^-v_1^+,
\end{equation}
which is an alternation between the arcs $\pi_i$ along the objects, and the upward vertical \emph{jumps}~$v_i^-v_{i+1}^+$ between them. 
As already said, there is a continuum of possible loops, representing different choices of the 
vertical lines (and thus points) at which we decide to jump from object to object, and 
of the paths along which the ``landing'' and ``take-off'' points are connected along 
each object.\footnote{%
  This is in stark contrast to the case of lines, studied in \cite{ArS1dcg}, where each 
  cycle corresponds to a \emph{unique} path of this kind, as long as the lines are in general position.}

To eliminate all cycles, it suffices to cut all the associated loops, in a manner made precise in the following easy lemma.
\begin{lemma} \label{cutpaths}
  For each $\Delta\in\T$, let $\Gamma_\Delta$ be a finite collection of algebraic curves drawn on~$\Delta$,
  and let $\OO_\Delta$ denote the collection of relatively open two-dimensional faces of $\A(\Gamma_\Delta)$;
  put~$\Gamma\coloneqq \bigcup_\Delta \Gamma_\Delta$ and $\OO\coloneqq \bigcup_\Delta \OO_\Delta$. Then,
  to verify that the depth relation among the pieces in $\OO$ is acyclic, it is 
  sufficient to ensure that, for each cycle $C$ in $(\T,\prec)$, and for each loop~$\pi\in\Pi(C,\T)$, 
  one of the (closed) arcs $\pi_i$ of $\pi$ has been cut by a curve in $\Gamma$.  Such cut may be performed along a subarc of some $\pi_i$.
\end{lemma}

\begin{remark} 
  Notice that we require that \emph{all} loops in $\Pi(C,\T)$
  be cut.  Indeed, for a specific loop~$\pi$, a subpath $\pi_i$ may be cut in
  such a way that it first leaves and then reenters the same face of~$\A(\Gamma_\Delta)$.  This by itself does not eliminate $C$, as $\pi_i$ can be replaced by a rerouted subpath~$\pi'_i$ that stays within the same face. However, replacing $\pi_i$ by $\pi'_i$ in $\pi$ produces a
  \emph{different} loop in~$\Pi(C,\T)$, which we also require to be
  cut.   The following proof handles this issue appropriately.
\end{remark}

\begin{proof}
  We proceed by contradiction: Assume that all loops
  in~$\Pi(C,\T)$, for every cycle~$C$ in~$(\T,\prec)$, have been cut,
  but nonetheless there remains a cycle
  $C' \colon o_1\prec o_2 \prec o_3 \prec \cdots \prec o_k \prec o_1$
  in~$(\OO,\prec)$.  In this case the set $\Pi(C',\OO)$ of loops realizing $C'$
  is non-empty, and we pick a loop $\pi\in \Pi(C',\OO)$, having
  the form \eqref{eq:pi}, where each subpath $\pi_i$ is contained
  in the corresponding piece~$o_i\in\OO$, and each vertical
  jump~$v_{i}^-v_{i+1}^+$
  moves from~$o_{i}$ to~$o_{i+1}$.  Each~$o_i$ is contained in some (not necessarily distinct) triangle
  $\Delta_i \in \T$, so $\pi_i$ is fully contained in $\Delta_i$,
  the jump $v_{i}^-v_{i+1}^+$ can be viewed as a vertical jump from $\Delta_i$
  to $\Delta_{i+1}$, and therefore
  $\Delta_1 \prec \Delta_2 \prec \cdots \prec \Delta_k \prec \Delta_1$ is a
  cycle in $(\T,\prec)$ with a witness loop $\pi$ that has not been cut, contradicting our assumption.
\end{proof}

\paragraph*{The polynomial partitioning.}
For a non-zero polynomial $f \in \reals[x,y,z]$,
of degree at most $D$, we let $Z(f)\coloneqq \{(x,y,z) \mid f(x,y,z)=0 \}$ denote its zero set. 
Removing $Z(f)$ from~$\reals^3$ creates $O(D^3)$ open connected \emph{cells}\footnote{%
  Note that they are not ``cells'' in the standard topological sense, being open and 
  potentially topologically non-trivial.  We will address this issue below.}
(see, e.g., Warren \cite{War}). The proposition stated below is a special instance, 
tailored to our needs, of the considerably more general result of Guth \cite{Gut}, 
which extends the earlier polynomial partitioning theorem of Guth and Katz \cite{GK2}.
\begin{prop}[Guth \cite{Gut}]
  \label{prop:gut}
  Given a set of $N$ lines in $\reals^3$ and an integer $1\leq D \leq \sqrt{N}$,
  there exists a non-zero polynomial $f \in \reals[x,y,z]$ of degree at most $D$, 
  so that each of the $O(D^3)$ cells of~$\reals^3\setminus Z(f)$ intersects 
  at most $cN/D^2$ of the given lines, for some absolute constant $c$.
\end{prop}

The first step of our construction resembles that for the case of lines in \cite{ArS1dcg}.
Specifically, let $\E$ denote the set of the $3n$ edges of the triangles in $\T$.
Let $f$ be a non-zero partitioning polynomial, of sufficiently large but \emph{constant} 
degree $D$, for the $3n$ lines supporting the segments of $\E$, as provided by Proposition~\ref{prop:gut}. 
That is, $\reals^3\setminus Z(f)$ consists of $k=O(D^3)$ open connected cells, each intersected 
by at most $cn/D^2$ (lines supporting) segments of $\E$, for an absolute constant $c>0$.\footnote{The constant $c$ here is three times the constant from Proposition~\ref{prop:gut}.}
We assume, without loss of generality, that $f$ is square-free.

\paragraph*{Preparing for the cycle elimination.}
The general strategy, similar to the one used in \cite{ArS1dcg}, is
to cut the triangles of $\T$ into pieces, using $Z(f)$, for a suitable 
partitioning polynomial $f$, in a manner detailed below, 
and then to recurse within each cell of the partition.

Each recursive step, at some node $\xi$ of the recursion tree, is associated with an open 
connected cell $\sigma_\xi$ and with a subset $\T_\xi^{(p)}\subseteq\T$, constructed as follows.
Let $\zeta$ denote the parent node of $\xi$, and let $\zeta_1={\rm root},\zeta_2,\ldots,\zeta_s=\zeta$
be the proper ancestral nodes of $\xi$, ordered from the root to $\zeta$.
(At the start of the recursion, the root has no proper ancestors; we take $\sigma_{\rm root}$ 
to be the entire 3-space, and $\T_{\rm root}^{(p)} = \T$.)
At each $\zeta_i$, we have constructed a partitioning polynomial~$f_{\zeta_i}$
of degree at most $D$. We put $F_\zeta := \Pi_{i=1}^s f_{\zeta_i}$, and
$\sigma_\xi$ is one of the open connected cells of~$\reals^3\setminus Z(F_\zeta)$,
that is contained in the parent cell~$\sigma_\zeta$.\footnote{%
  By construction, each $\sigma_\xi$ is either fully contained in $\sigma_\zeta$ or disjoint from it.}
As the recursion depth is only $O(\log_D n)$, the degree of $F_\zeta$ is at most
$O(D\log_D n)$.  $\T_\xi^{(p)}$ is the set of all triangles $\Delta\in \T$
that \emph{pierce}~$\sigma_\xi$. These are the triangles that have at least 
one edge intersecting~$\sigma_\xi$.

We process $\xi$ (and $\sigma_\xi$) as follows. Put $n_\xi := |\T_\xi^{(p)}|$.
We construct a partitioning polynomial~$f_\xi$, of degree at most $D$,
for the set of the $3n_\xi$ lines that support the edges of the triangles 
in $\T_\xi^{(p)}$, as in Proposition~\ref{prop:gut}, form the product polynomial
$F_\xi := F_\zeta f_\xi$ (at the root, we just put $F_\xi := f_\xi$), and
collect all the open connected cells of $\reals^3\setminus Z(F_\xi)$ that
are \emph{contained in $\sigma_\xi$}. (As just noted for $\zeta$, any cell of
$\reals^3\setminus Z(F_\xi)$ is either fully contained in $\sigma_\xi$ or is 
disjoint from~it.) We create a child node $\eta$ of $\xi$ for each
such cell $\sigma$, generate a recursive subproblem at $\eta$, and
take $\sigma_\eta$ to be $\sigma$ and $\T_\eta^{(p)}$ to be the set of all
triangles of $\T$ that pierce~$\sigma$. Actually, $\T_\eta^{(p)}$ is also
the set of all the triangles of $\T_\xi^{(p)}$ that pierce
$\sigma$. This follows from the easy observation that if $\Delta$ is
piercing at some node, it must also be piercing at the parent node.
(In counterpositive terms, if $\Delta$ slices the cell of some recursion
node, it also slices the cells of all its descendants.)
We note, though, that algorithmically it is more efficient to extract the piercing triangles at $\eta$ from those that are piercing at its parent $\xi$.

The recursion terminates at nodes $\xi$ for which $n_\xi \le D^2/c$, where $c$
is three times the constant in Proposition~\ref{prop:gut}.

Before branching into the recursive child steps, we draw (that is, generate) curves on 
the triangles of $\T_\xi^{(p)}$ (other triangles that meet $\sigma_\xi$ but do not 
pierce it are ignored in this part of the procedure). 
We generate a small number of constant-degree algebraic curves
that result from certain interactions of $f_\xi$ with these triangles, in a manner
to be detailed below. Each curve that we draw is clipped to within the closure
of $\sigma_\xi$; that is, we only maintain the maximal connected arcs of the
intersection of the curve with that closure. The number of such subarcs and 
their pattern of intersection will be examined in Section~\ref{sec:drawing}. 

We now proceed to describe the process in full detail.
As in \cite{ArS1dcg}, define the \emph{level} $\lambda(q)$ of a point~$q\in\reals^3$ 
with respect to~$Z(f)$ (where, as above, $f=f_\xi$) to be the number of intersection 
points of $Z(f)$ with the relatively open downward-directed $z$-vertical ray $\rho_q$ 
emanating from~$q$.\footnote{%
  The polynomials constructed at the proper ancestral nodes of $\xi$ do not matter here,
  since their zero sets are disjoint from $\sigma_\xi$, and in the processing of $\xi$ 
  we only cater to cycle-realizing loops that are fully contained in $\sigma_\xi$.}
Formally, if~$q=(x_0,y_0,z_0)$, we consider the univariate polynomial
$f_0(z) = f(x_0,y_0,z)$, and the level~$\lambda(q)$ of~$q$ is the number of 
real zeros of~$f_0$ in~$(-\infty,z_0)$, counted with multiplicity. 
At points $q$ where the entire vertical line through $q$ is contained in $Z(f)$, 
so that $f_0\equiv 0$, $\lambda(q)$~is undefined.  We will explicitly deal with 
such points in the analysis below.  The number of such lines is~$O(D^2)$ (see below), 
unless $Z(f)$~contains a ``vertical curtain,'' i.e., if $f$~has a factor that does not depend on~$z$.

\subsection{The procedure for cutting the triangles}
\label{sec:cutting}

The procedure is recursive. At each step $\xi$ of the recursion we have the subset $\T_\xi^{(p)}$ 
of the triangles that pierce the corresponding cell $\sigma_\xi$. 
We rename this set as $\T_\xi$, to simplify the notation, and process it as follows. 

Assume first that $|\T_\xi| > D^2/c$, where $c$ is the constant defined above.
We apply the following steps.

\begin{enumerate}[(a),wide,labelindent=0pt]
\item 
We construct a partitioning polynomial $f = f_\xi$, as in Proposition~\ref{prop:gut}, 
for (the lines supporting) the edges of the triangles of~$\T_\xi$; the degree~$D=D(\eps)$ 
of~$f$ is a sufficiently large constant that depends only on the prespecified $\eps$. 
The same value of $D$ is used at all levels of recursion. We put $F = F_\xi = F_\zeta f$, 
where $F_\zeta$ is the product polynomial defined above for the parent node $\zeta$,
and $F_{\rm root} = 1$.

\item 
We generate curves on the triangles of~$\T_\xi$, with up to $O(D^h\log^2n)$ curves,
each of degree up to~$O(D^2)$, on each triangle, for a suitable (small) absolute constant $h>0$.

\item 
We recurse within each cell~$\sigma$ of $\reals^3\setminus Z(F)$ that is contained
in $\sigma_\xi$, with the subset of the triangles of $\T$ (or, as noted above, of $\T_\xi$)
that pierce~$\sigma$.
\end{enumerate}

The processing of nodes $\xi$ with $|\T_\xi| \le D^2/c$ consists of applying the BSP
technique of \cite{PY} to $\T_\xi$, and clipping each of the resulting cutting
segments to within $\sigma_\xi$; see below for more details.

We now spell out the details of step~(b). The curves that we draw are of the 
following three types. (Since we only draw curves within $\sigma_\xi$, it 
suffices, in some of the substeps, to consider $f=f_\xi$ rather than $F=F_\xi$. In what follows,
this clipping to within $\sigma_\xi$ is mostly implicit, but the reader should bear in mind that it
does take place wherever applicable.)

\begin{enumerate}[(i),font=\bf,wide,labelindent=0pt]

\item \emph{Traces}:
For each triangle $\Delta\in \T_\xi$, not fully contained in $Z(f)$, we draw~$Z(f)\cap \Delta$ 
on~$\Delta$.  We call this the \emph{trace of~$f$} on~$\Delta$.  It is a curve of degree at most~$D$.

If $\Delta \subset Z(f)$, that is, if the plane $h_\Delta$ supporting $\Delta$ is a component 
of $Z(f)$, we do not draw any such curve on $\Delta$. Note that the traces just defined include 
the non-empty segments~$\Delta'\cap h_\Delta$ for the other 
triangles~$\Delta' \in \T_\xi \setminus \{\Delta\}$, and non-empty portions of 
some of these segments may show up, and will then be drawn, on the respective triangles~$\Delta'$.

\item \emph{Critical shadows}:
We next consider the set of points $p\in Z(f)$ that are either singular or at which 
$Z(f)$ has a~$z$-vertical tangent line. This set is contained in the common zero set 
$S=S(f) \coloneqq Z\left(f,\frac{\partial f}{\partial z}\right)$ of $f$ and 
$\frac{\partial f}{\partial z}$, so, for simplicity, we use $S$ instead. Intuitively, $S$ consists of two parts: 
if $Z(f)$ contains a ``vertical curtain,'' i.e., a set of the form $G \times \reals$, 
for some (maximal one-dimensional) $G \subset \reals^2$ (equivalently, if $f$ has a factor that 
does not depend on $z$), then $G \times \reals \subset S$, and the 
remainder of $S$ is at most one-dimensional.

Let $H$ denote the vertical curtain spanned by $S$, namely, the union of all 
vertical lines that pass through points of $S$. Since $S$ is an algebraic variety 
of degree $O(D^2)$ (and the only two-dimensional portions of $S$, if any, are vertical curtains), 
$H$~is a two-dimensional variety of the same degree (see, e.g., \cite{Fu84}).\footnote{%
    The equation for $H$ is obtained by eliminating $z$ from the system $f=\frac{\partial f}{\partial z} = 0$.
    Note that $H$ includes all the vertical components of $Z(f)$.}

  We then draw, on each triangle $\Delta\in \T_\xi$, including triangles contained in $Z(f)$, 
  the \emph{critical shadow} curve $H \cap\Delta$.  It is a curve of degree~$O(D^2)$.

\item 
\emph{Wall shadows}:
We eventually want to proceed recursively, within each cell $\sigma$ of the partition, 
but before doing so, we discuss in more detail the notions of piercing and slicing 
triangles, as already defined, to set the stage for the last type of curves that we draw.

Consider the interaction of the triangles $\Delta \not\subset Z(f)$ with the cells of 
the partition. By construction, each cell $\sigma$ meets the edges of at most $O(n/D^2)$ triangles. 
However, $\Delta \in \T_\xi$ may meet $\Theta(D^2)$ cells in the worst case, a consequence of Warren's theorem \cite{War}. 
Therefore, each cell $\sigma$ meets $O(n/D^2)$ triangle edges and, on average, $O(nD^2/D^3) = O(n/D)$ 
triangle interiors; in the worst case, the latter average bound is tight.
Roughly speaking, the latter quantity is too large and yields an unfavorable recurrence, so we have to be more careful:
As in the introduction, and earlier in this section (refer to Figure~\ref{fig:pierce}), we say that a triangle $\Delta$ 
\emph{pierces} $\sigma$ if one or more of its edges intersects (the open cell) $\sigma$, 
and that it \emph{slices} $\sigma$ if $\Delta$~meets~$\sigma$, but its edges do not.
Our plan now is to recurse, for each cell~$\sigma=\sigma_\eta$, where $\eta$ is a child of~$\xi$,
only on the set $\T_\eta = \T_\eta^{(p)}$ of triangles (of $\T$) that pierce it, 
and disregard, for the purposes of recursion, the triangles of $\T_\xi$ that slice it.  
However, this is safe only if the slicing triangles do not participate in any uncut loop 
that represents  a depth cycle and is fully contained within $\sigma$.  We draw additional curves on the 
triangles, as described next, to ensure this. 

We construct the \emph{vertical decomposition} $\VV\coloneqq\VV(Z(F))$ of $3$-space 
produced by $Z(F)$; see Appendix~\ref{sec:vertical-decomp} for details on vertical 
decompositions of the kind we use here, and see, e.g., \cite{SA} for more general information.
As detailed in the Appendix, $F$ is a product of $k := O(\log_D n)$ polynomials, 
each of degree at most $D$, so we actually construct the vertical decomposition 
of the \emph{arrangement}~$\A(\F)$ of the collection~$\F$ of the zero sets of these $k$~polynomials.
We retain only the portion of~$\VV$ within the closure of~$\sigma_\xi$; 
by the properties of vertical decomposition (and since we construct it for $F$ rather than for $f$)
it is easy to verify that each three-dimensional cell of $\VV$ is either
fully contained in $\sigma_\xi$ or disjoint from it.
Let the \emph{2-skeleton} $\VV^{(2)}$ of a vertical decomposition $\VV$ be the complement 
of the union of its open three-dimensional cells.
As just said, in what follows we only consider the portions of $\VV$ and
of $\VV^{(2)}$ within the closure of $\sigma_\xi$; 
we continue to denote these portions as $\VV$ and $\VV^{(2)}$.

We are now ready to draw our third kind of curves, called \emph{wall shadows}, on the triangles of $\T_\xi$.  
Fix $\Delta \in \T_\xi$. For every open three-dimensional cell $\nu$ of $\VV$ meeting $\Delta$, 
we draw the one-dimensional boundary of $\nu \cap \Delta$ on $\Delta$.
It is also possible for the closure $\bar\nu$ of an open three-dimensional cell $\nu$ 
of $\VV$ to meet $\Delta$ without the cell itself meeting $\Delta$. 
This happens when $\Delta\subset Z(f)$ is part of the floor or ceiling of $\nu$. 
In this case, we draw the boundary of the floor/ceiling of $\nu$ on $\Delta$.
Other scenarios where $\Delta$ meets $\bar{\nu}$ but not $\nu$ arise when
$\Delta\cap\bar{\nu}$ is $0$- or $1$-dimensional. In these cases $\nu$ does not cause any curves
to be drawn on $\Delta$ (but neighboring cells will generate such curves).

\begin{remark}
An earlier version of this work \cite{triangle-cycles-soda} devoted significant 
effort to handling the fact that cells arising from a polynomial partitioning 
may have non-trivial topology. Without additional care, we could have, say, a donut-shaped 
cell containing (a loop representing) a cycle formed solely from triangles that 
slice the cell \cite{triangle-cycles-soda}; see Figure~\ref{fig:domino}.
Eliminating such cycles required a considerable portion of the analysis in 
\cite{triangle-cycles-soda}, by a controlled additional subdivision of the triangles,
made to prevent precisely this from happening.  In the current approach, we subdivide all 
triangles (in the present step (iii)) according to their interaction with the vertical decomposition~$\VV$, thereby 
guaranteeing, essentially, that any uncut geometric realization (loop) of a cycle in the depth relation 
has to be contained in a single cell of $\VV$ (which, by construction, has trivial topology), 
or in $\VV^{(2)}$, a situation that is easy to handle and is detailed below. This comes
at the expense of possibly somewhat increasing the number of curves on each triangle (with 
no real asymptotic increase).
\end{remark}

\begin{figure}
    \centering
    \includegraphics[width=0.6\textwidth]{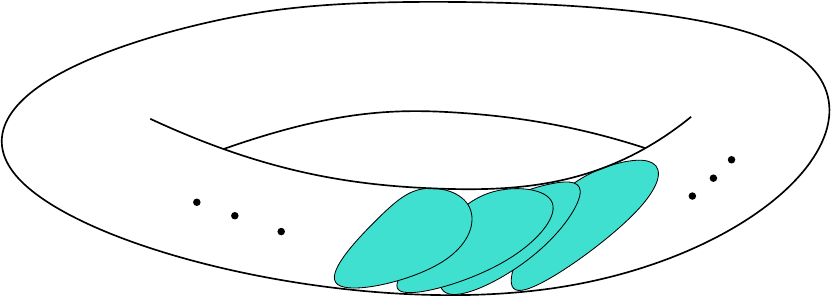}
    \caption{A donut-shaped cell is sliced by several triangles in a pinwheel fashion. 
    These triangles form a depth cycle, so that at least one loop realizing it stays 
    completely within the cell.  Only four intersections of the cell with these triangles 
    are shown, out of many.}
    \label{fig:domino}
\end{figure}

\paragraph{Implications.}
We note that the vertical decomposition $\VV=\VV(Z(F))$ is coarser than the 
\emph{cylindrical algebraic decomposition}~(CAD) \cite{collins75,SS83} of $Z(F)$. 
As described in the Appendix, $\VV$~is based on a collection of curves
drawn on the surfaces of $\F$, which are the intersection curves of pairs of the 
surfaces, and the loci of singular points and points with $z$-vertical tangencies 
on the individual surfaces. There is a total of $O(k^2) = O(\log^2n)$ such curves,
each of degree $O(D^2)$. These curves define the vertical walls
of the first-stage prisms of~$\VV$, and these walls are thus vertical surfaces 
of degree $O(D^2)$ too. The second stage of the decomposition creates additional
\emph{planar} vertical walls, each contained in a vertical plane orthogonal to 
the $x$-axis. The resulting wall shadows that these walls produce on our
triangles are straight segments. It is easy to obtain, from the construction 
of $\VV$, the (somewhat crude) upper bound $O(k^4D^4) = O(D^4\log^4n)$
on the number of these walls, implying a bound of $O(nD^4\log^4n)$ on 
the number of these segments.

To recap, we draw, on each triangle of $\T$, at most $O(\log^2n)$ curves 
of degree $O(D^2)$ and $O(D^4\log^4n)$ straight segments are drawn in step~(iii).

Let $\nu$ be an open three-dimensional cell of $\VV$ and let $\Delta \in \T$ be a triangle 
that slices~$\nu$ (that is, as for the undecomposed cells of $\reals^3\setminus Z(f)$, 
$\Delta$ meets $\nu$, but the boundary of $\Delta$ avoids~$\nu$); note that $\Delta$
may or may not belong to $\T_\xi$. Consider a loop~$\pi$
of the form~\eqref{eq:pi} arising from some depth cycle among portions of the original 
triangles (only loops fully contained in $\sigma_\xi$ need to be considered). 
Recall that such a loop consists of vertical upward jumps alternating with 
connected arcs lying on triangles of~$\T$. We say that $\pi$ \emph{visits~$\Delta$ inside 
$\nu$} if its intersection with~$\Delta\cap\nu$ is non-empty. Notice that $\pi$ necessarily 
arrives at $\Delta$ from a point lying below the plane $\pi_\Delta$ supporting $\Delta$ 
and leaves $\Delta$ to a point above that plane.
(This visit might consist of a single point where some unrelated vertical jump just crosses $\Delta$.)

The following lemma is crucial for the correctness of our cycle-cutting procedure; we defer 
its rather technical topological proof to Appendix~\ref{sec:where-we-prove-it}.  
See the proof for an explanation of why working with bounded cells suffices.
\begin{lemma}
  \label{lem:no-disconnect-really}
  In the above terminology, a loop \(\pi\) fully contained in a bounded open 
  cell~\(\nu\) of~$\VV(Z(F))$ (within $\sigma_\xi$) cannot visit a triangle $\Delta$ that slices $\nu$.
\end{lemma}

Before continuing with the generation of curves to be drawn on the triangles, we make 
the following observations that will be very helpful in proving the correctness of our procedure.

\begin{lemma}
  \label{lemma:jump}
  Consider a loop~$\pi$ of the form \eqref{eq:pi} that is contained in $\sigma_\xi$ and
  represents some cycle $C$ in $\T_\xi$, such that one of its 
  open vertical jump segments meets $Z(f)$. Then $\pi$ is cut.
\end{lemma}

\begin{proof}
Assume, without loss of generality, that the segment in question is $s\coloneqq v_1^-v_2^+$.
We traverse $\pi$ in circular order, starting at $v_1^-$, say. First, assume that 
$s$ is not fully contained in $Z(f)$, so the level $\lambda(\cdot)$ is defined on $s$.
At each point where $s$ intersects $Z(f)$, the level increases as we traverse~$s$ 
(upwards, as $\pi$ does) past this point. To summarize, if $s$ is not contained in $Z(f)$,
at some point of $\pi$ the level goes up.

Since (a)~at every vertical jump $v_i^-v_{i+1}^+$ the level can only stay the same or increase
(or be undefined), but never decrease, and (b)~$\pi$ is a closed loop, 
the level has to come back to its original value or become undefined, 
at some point $q$ on one of the subpaths $\pi_i$ of $\pi$ (including possibly 
an endpoint of such a path). However, $\pi_i \subset  \Delta_i\in\T_\xi$, 
and we cut $\Delta_i$ at each point where $Z(f)$ meets $\Delta_i$ (in step~(i)), 
and along the critical shadow curves (in step~(ii)), which, collectively, are precisely 
the points where such an event may occur, completing the 
argument for the case where the level is defined on $s$.

If $s\subset Z(f)$, the level is undefined on $s$.  On the other hand, $v_1^-$ and $v_2^+$ 
are critical shadow points, and $\pi$ is cut at both of them. A similar argument
applies if $\lambda$ becomes undefined at any other point along $\pi$. 
This completes the proof.
\end{proof}

\begin{lemma}
  \label{lemma:pi-i}
  Consider a depth cycle $C$ among the triangles of $\T_\xi$ and a loop $\pi\in\Pi(C,\T_\xi)$ 
  contained in $\sigma_\xi$, of the form \eqref{eq:pi}, such that one of its subpaths~$\pi_i$ 
  meets $\VV^{(2)}$. 
  Then $\pi$ is cut.
\end{lemma}

\begin{proof}
  By definition of $\pi$, $\pi_i \subset \Delta_i$, for some triangle $\Delta_i \in \T_\xi$. We consider several ways in which $\pi_i$ may meet $\VV^{(2)}$.  We subdivide $\VV^{(2)}$ into the relatively open floors and ceilings of cells of~$\VV$ and the relatively closed vertical walls of these cells; these two sets are not necessarily disjoint.

  If $\pi_i \subset \Delta_i$ intersects the vertical walls, then it is cut at such a point or points by a wall shadow drawn on~$\Delta_i$ in
step~(iii) of our construction.  (Recall that we assumed that there are no vertical triangles in $\T$, so $\Delta_i$ cannot overlap a vertical wall.)
Note that $\pi_i$ may partially or completely overlap a wall shadow, but it is still cut in such cases because the drawn curves are removed in the 
  formation of the final triangle pieces.

  Now suppose $\pi_i \subset \Delta_i$ avoids the relatively closed vertical walls of cells in $\VV$ and therefore only meets the intersection of a
relatively open floor of such a cell with the relatively open ceiling of another cell (lying immediately below the first one).  Recall that 
the non-vertical surfaces of $\VV^{(2)}$ (within the region $\sigma_\xi$ of the current recursive instance)
 are contained in $Z(f)$. 

If $\Delta_i$~is not contained in $Z(f)$ then it is cut in step~(i). In particular, $\pi_i\subset \Delta_i$ must
be cut at a point of~$\Delta_i\cap Z(f)$.

  Otherwise, the plane $h_{\Delta_i}$ supporting $\Delta_i$ is fully contained in $Z(f)$. Then we made no cuts on $\Delta_i$ in step~(i) of our
construction, and $\pi_i$ is fully contained in the relatively open floor of a cell $\nu$ of $\VV$.  Consider following $\pi$ from $\pi_i$ onward.  The
vertical jump $v_i^-v_{i+1}^+$ following $\pi_i$ must enter the open cell $\nu$ and therefore leave $Z(f)$.  In particular,
$\lambda(v_i)<\lambda(v_{i+1}^+)$ (the level cannot be undefined here, as is easily checked, as $v_i^-v_{i+1}^+$ passes through the open cell~$\nu$
of~$\VV$).  As in the preceding proof, since
 the level comes back to its value after a circular traversal of the loop $\pi$, it must go back down or become undefined somewhere along
the loop.  The level can only go up along vertical jumps (unless it is undefined there). It thus follows, as before,
that some of the cuts made in steps~(i) and~(ii) of our construction, which
eliminate precisely those points on the triangles of $\T_\xi$ where the level changes or becomes undefined, will cut $\pi$, as desired.

  Finally, if the level is undefined along an open vertical jump, it's also undefined at its endpoints,
 which lie on triangles of $\T_\xi$, so $\pi$ is also cut at such points.

  Having exhausted all cases, we have completed the proof.
\end{proof}

\item \emph{Recursively constructed curves}:
We finally apply recursion within each cell~$\sigma$ of~$\reals^3\setminus Z(F_\xi)$
contained in $\sigma_\xi$. We generate a child $\eta$ of $\xi$, put
$\sigma_\eta = \sigma$, and pass to this recursive step the subset~$\T^{(p)}_\eta$ 
of all the triangles in $\T$ that pierce~$\sigma$.\footnote{%
  Note that the vertical decomposition $\VV$ is used for drawing curves on the triangles, but \emph{not} for guiding the recursion.  Informally, $\VV$ has in general 
  too many cells to yield a favorable recurrence relation.}
Recall that $\sigma$, as a spatial entity, is also intersected by additional slicing 
triangles, which will not be considered in the recursive subproblem; an implication of
Lemma~\ref{lem:no-disconnect-really} allows us to completely ignore these triangles.
See below for more details.

As in the case of lines in \cite{ArS1dcg}, and as already mentioned, the bottom of 
the recursion is at nodes $\xi$ for which $|\T_\xi^{(p)}| \le D^2/c$, where $c$~is the 
(modified) constant in Proposition~\ref{prop:gut}. For such cells we apply the Paterson-Yao
binary space partitioning \cite{PY}, which cuts the triangles into $O(|\T_\xi^{(p)}|^2)=O(D^4)$
triangular pieces, whose depth relation does not contain cycles, and retain only 
the portions of those pieces within $\sigma_\xi$. Following our strategy, we do not 
really perform the cuts yet, but just add the straight segments (clipping them,
as needed, to within~$\sigma_\xi$) bounding the triangular pieces
to the collections of curves on the triangles.
\end{enumerate}

\subsection{All cycles are eliminated}

Let $\Gamma$ denote the set of all (clipped) curves that have been generated throughout the recursion.
We write $\Gamma$ as the disjoint union $\bigsqcup_\Delta \Gamma_\Delta$, 
where $\Gamma_\Delta$ is the set of curves drawn on~$\Delta$, for each~$\Delta\in\T$.

Consider a node $\xi$ of our recursive procedure, and let $\zeta$ be its parent node
(ignore $\zeta$ when $\xi$ is the root).
Recall that each curve $\gamma\in\Gamma$, constructed when partitioning the cell~$\sigma_\zeta$ 
corresponding to node~$\zeta$, is clipped to within~$\sigma_\zeta$, and may visit several 
subcells of that cell. The forthcoming analysis also covers the case where $\xi$ is the
root, and is only simpler then.

\begin{lemma} \label{tricutok}
The procedure described above eliminates all the depth cycles in $\T$, in the sense that, 
for each cycle \(C\) in $(\T,\prec)$ and for each loop $\pi\in\Pi(C,\T)$ of the form \eqref{eq:pi}, 
at least one of the ``on-triangle'' closed subpaths $\pi_i$ of \(\pi\) meets a curve of \(\Gamma\).
\end{lemma}

\begin{proof}
Let $C \colon \Delta_1\prec \Delta_2\prec\cdots\prec \Delta_k\prec \Delta_1$ 
be a cycle in $(\T,\prec)$, and let $\pi$ be a loop in $\Pi(C,\T)$ of the form \eqref{eq:pi}.
Let $\T_C=\{\Delta_i \mid i=1,\dots, k \}$ be the set of triangles appearing in~$C$.  
Let $\xi$~be the lowest (farthest from the root) node that satisfies the
following two properties: (a)~$\pi$~is completely contained in the cell 
$\sigma_\xi$ corresponding to $\xi$ in the subdivision formed at its parent node $\zeta$
(recall that $\sigma_{\rm root} = \reals^3$), and 
(b)~all the triangles in $C$ are piercing for $\zeta$ (this condition is 
vacuous when $\xi$ is the root). Such a node $\xi$ always exists and is unique.
Indeed, the root satisfies~(a) and (vacuously)~(b), and the set of nodes satisfying both~(a) and~(b) is easily seen to form a 
(contiguous) path: $\xi$~is the bottommost node of this path.  We will ensure that $\pi$ is cut
when processing either $\xi$ or its parent $\zeta$, depending on which of the
following two situations arises.

Each triangle $\Delta_i$ meets $\sigma_\xi$, so it is either a piercing triangle or a 
slicing triangle for $\sigma_\xi$. We distinguish between two cases: 
(i)~All triangles in $\T_C$ pierce $\sigma_\xi$.
(ii)~At least one of these triangles slices $\sigma_\xi$. 

Consider first the situation in case (ii). Assume that $\xi$ is not the root (the
analysis only gets simpler when $\xi$ is the root), and let $\zeta$ denote its parent.
In this case we do not consider the decomposition of $\xi$,
but only the decomposition done at $\zeta$.
The strategy of the proof is to show that, in case (ii), one of the subpaths $\pi_i$
of $\pi$ is cut by one of the curves drawn while processing $\zeta$.
To carry out this strategy, we note that $\pi$ cannot be fully contained in 
any (open three-dimensional) cell $\nu$ of the vertical decomposition
$\VV = \VV(Z(F_\zeta))$ within $\sigma_\xi$, because then each triangle 
in $\T_C$ that slices $\xi$ also slices $\nu$ (and there is at least one such
triangle, by assumption), so $\pi$ visits this triangle in $\nu$,
but Lemma~\ref{lem:no-disconnect-really} asserts 
that this is impossible. It therefore must be the case that either $\pi$ meets 
an open three-dimensional cell~$\nu$ of $\VV$ but is not fully contained in~$\nu$, 
or $\pi$ is fully contained in the boundary portion~$\VV^{(2)}$ of~$\VV$.
In the latter case, all of the subcurves $\pi_i$ of $\pi$ lie in $\VV^{(2)}$ and therefore $\pi$ is cut, by Lemma~\ref{lemma:pi-i} applied to $\zeta$.
Hence $\pi$ must meet $\nu$ without being contained in it, and without any of its subpaths~$\pi_i$ meeting $\VV^{(2)}$, 
which is implied by another application of Lemma~\ref{lemma:pi-i}.  This means that one of the open vertical jumps on $\pi$ must meet $\VV^{(2)}$
without meeting $Z(f)$. (If it did meet $Z(f)$, Lemma~\ref{lemma:jump} would imply that $\pi$ is cut.)

We now argue that such a jump cannot exist: By assumption it has to meet an open cell of $\nu$ and therefore cannot be fully contained in a vertical feature of $\VV^{(2)}$.  So a point at which it enters or leaves $\VV^{(2)}$ must belong to a non-vertical feature of $\VV^{(2)}$ and those are fully contained in $Z(f)$, leading to a contradiction.

To recap, in case~(ii) $\pi$ has been cut during the nonrecursive processing of 
the parent $\zeta$ of~$\xi$. (Note that case (ii) cannot arise when $\xi$ is the root.)

Let us now consider case (i), in which all the triangles of $\T_C$ pierce $\sigma_\xi$.
This is where we do consider the partitioning at $\xi$.
If $\xi$ is a leaf, the claim holds because the BSP constructed at $\xi$ eliminates 
all cycles among the triangles of $\T_C$ (in the sense of Lemma~\ref{cutpaths}).  
To make this argument rigorous, we need to bear in mind that we only retain the
portions of these cuts within $\sigma_\xi$. However, since $\pi$ is fully
contained in $\sigma_\xi$, any point at which an unclipped cutting segment meets
(some subpath $\pi_i$ of) $\pi$ necessarily lies in $\sigma_\xi$, so it lies
on the (retained) clipped portion of that segment.

If $\xi$ is not a leaf, we constructed a partitioning polynomial $f_\xi$ for $\T_\xi^{(p)}$,
and used $F_\xi = F_\zeta f_\xi$ (where $\zeta$ is the parent of $\xi$, or else, 
when $\xi$ is the root, $F_\zeta = 1$) to partition $\sigma_\xi$ into subcells.
If $\pi$ is not fully contained in any subcell of $\sigma_\xi$, then it has 
to meet (and possibly partially overlap) $Z(f_\xi)$. If $\pi$ meets $Z(f_\xi)$ at a point or points that lie on one of the subpaths $\pi_i$, then $\pi$
is cut (in step~(i)) (in the case of overlap, all or some subarc of $\pi_i$ might be removed).
The only remaining case is that some open vertical jump,
say $v_1^-v_2^+$, of $\pi$ meets $Z(f_\xi)$. The analysis is then entirely analogous to that of the proof of Lemma~\ref{lemma:jump},
and implies that $\pi$ is cut in this case too.

Finally, the case where $\pi$ is fully contained in some subcell $\sigma_{\xi'}$
of $\sigma_\xi$ is impossible, as it contradicts the definition of $\xi$ as the
lowest node satisfying both properties (a) and (b) at the beginning of the proof
(namely, that $\pi$ is contained in the cell and that all its triangles pierce
the parent cell). 

Having covered all possible cases, the lemma follows.
\end{proof}

This finishes the proof of correctness of our procedure. 
We still need to fill two gaps:
(i)~We need an upper bound on $|\Gamma|$, the number of clipped curves that
the procedure generates. 
(ii)~We need to cut each triangle into pieces of constant description complexity and control their number.
We now proceed to describe each of these steps in detail.

\subsection{Bounding the number of curves}
\label{sec:number}

Consider the situation at some recursion node $\xi$, and note that the clipped
curves generated at $\xi$ are obtained essentially in two stages.
Using an upper bound that is certainly too crude, but suffices for our purposes,
we first generate on each triangle of $\T$, in steps (i)--(iii), up to 
$O(D^h\log^4n)$ algebraic curves, each of degree at most $O(D^2)$. 
Then we intersect these curves with the closure of $\sigma_\xi$
to obtain their clipped versions. Ignore for the moment the second stage,
and define $\chi(\T)$, for $\T = \T_\xi$, to be the maximum number of
(unclipped) curves that our procedure generates on the triangles of $\T$, 
for the fixed choice of $D$ that we use throughout the recursion.\footnote{%
  Note that the actual number of curves depends on the partitioning polynomials 
  constructed throughout the recursion, and $\chi(\T)$ maximizes this over all 
  possible choices of partitioning polynomials of degree at most~$D$.}
Put $\chi(n) \coloneqq \max_{|\T|=n} \chi(\T)$, where the maximum is taken 
over all collections~$\T$ of $n$ non-vertical pairwise disjoint relatively 
open triangles in~$\reals^3$. Then $\chi(\T)$ satisfies the following 
recurrence relation (for $|\T|> D^2/c$)
\[
\chi(\T) \le bD^3\chi(c|\T|/D^2) + O(|\T| D^{h}\log^4|\T|) ,
\]
where $b$, $c$, and $h$ are suitable absolute constants. The overhead term 
$O(|\T|D^{h}\log^4|\T|)$ comes from the (crude) bound on the number of wall shadows, 
$O(D^h\log^4|\T|)$ per triangle, drawn in step~(iii) of the construction, which
dominates the number of all other curves non-recursively constructed at the present node.

Maximizing over $\T$ produces the recurrence
\[
  \chi(n) \le
  \begin{cases*}
    bD^3\chi(cn/D^2) + O(nD^{h}\log^4n) , & for $n>D^2/c$ \\
    O(D^4), & for $n\le D^2/c$ .
  \end{cases*}
\]
As is easily verified, the solution of this recurrence is 
$\chi(n) = O(n^{3/2+\eps})$, for any $\eps>0$, provided that we choose~$D$ so as to satisfy 
\[
  D^{2\eps} \ge 2bc^{3/2+\eps}.
\]
That is, when $\eps$ is prescribed, we need to choose $D=2^{\Theta(1/\eps)}$, with a 
suitable constant of proportionality. Conversely, with $D$ as the specified parameter 
(that is, with an explicit control over the degree of the curves that we are willing to draw), 
we have $\eps = O(1/\log D)$.

\subsection{The complexity of the arrangements of the clipped curves}
\label{sec:drawing}

To recap, the various steps of the construction generate a collection of curves 
on the triangles. Altogether $O(D^h\log^4n)$ curves, each of degree at most $O(D^2)$,
are generated for each piercing triangle at each recursive level, the majority 
of which are those drawn in step~(iii). However, for each curve $\gamma$ we 
retain only its portion within the cell at which $\gamma$ was generated.

Upon termination of the entire recursive process, we take each triangle $\Delta\in\T$, 
and consider the planar map~$M_\Delta$ formed on~$\Delta$ by the hierarchy of curves
constructed for $\Delta$. That is, we take each curve~$\gamma$, generated at some 
recursive node~$\xi$ where $\Delta$ was a piercing triangle, clip $\gamma$
to within the cell $\sigma_\xi$, and draw only the clipped portion
$\gamma\cap\sigma_\xi$; we repeat this operation for all triangles $\Delta$
and all recursive steps~$\xi$.
 
Each vertex of $M_\Delta$ is either (a) an endpoint of a connected component of 
the clipped portion of some curve~$\gamma$, or (b) an intersection point between two 
(clipped) curves $\gamma$, $\gamma'$, such that either (b.i) both arcs are generated at
the same recursive step, within the same cell $\sigma$, or (b.ii) up to a swap between the arcs,
$\gamma$, $\gamma'$ are generated within two respective cells $\sigma$, $\sigma'$, such that
the step that generated~$\sigma'$ is a proper ancestor of the step that generated~$\sigma$.
These properties follow easily from the hierarchical nature of our drawings.

The number of clipped connected subarcs, over all the triangles, is at most the number of unclipped
curves, which we have shown to be $O(n^{3/2+\eps})$, with a suitable constant of 
proportionality, plus the number of cuts
that the clipping creates. Any such cut, of some curve $\gamma$ generated at some step $\xi$, 
occurs where $\gamma$ crosses the boundary of $\sigma_\xi$, and, in particular, at a point of\footnote{%
  In general, $\gamma$ crosses several cells into which $\sigma_\xi$ is split by $Z(f_\xi)$, but these
  points are not considered as endpoints of subarcs of $\gamma$.}
$\gamma\cap Z(F_\xi)$. Since $\gamma$ is a planar algebraic curve of 
degree at most $O(D^2)$, and $Z(F_\xi)$ is
an algebraic surface of degree $O(D\log_D n)$, it follows from B\'ezout's theorem that the
number of pieces into which $\gamma$ is cut is $O(D^3\log_D n)$.

The exceptions are endpoints of curves that lie on the edges of the corresponding triangles;
we may ignore these vertices, as we have only $O(D^2)$ such points for each of the curves 
that we draw, as is easily checked.

We next bound the number of intersection points of clipped arcs with other (clipped) 
arcs constructed at (proper and improper) ancestral recursive steps.
For each arc $\gamma$, formed along some triangle~$\Delta$, within a cell~$\sigma$ at 
some recursive step, the number of the ancestral cells of $\sigma$ is $O(\log_D n)$, 
and each of them generates on $\Delta$ up to $O(D^h\log^4n)$ curves of degree at most $O(D^2)$. 
For the present argument, treat these curves as drawn in their entirety---this will only 
increase the number of intersection points on $\gamma$. Since $\gamma$ is one of these curves, 
the number of intersection points of $\gamma$ with any other curve is $O(D^4)$, which is a 
consequence of B\'ezout's theorem.  It follows that the number of vertices that can be formed 
along $\gamma$ is at most $O(D^{h+4}\log^5_D n)$, which is quite possibly a gross overestimate, 
but we do not attempt to optimize it.\footnote{%
  This slack is indeed quite generous, but it only applies to curves generated at a pair 
  of nodes, one of which is an ancestor of the other. What we do not want to pay for
  are intersections between curves generated at two unrelated nodes.}

Adding these two bounds, and multiplying by the number of curves, as provided in 
Section~\ref{sec:number}, we conclude that the overall complexity of 
the maps $M_\Delta$, over all triangles $\Delta$, is 
\[
  O(D^{h+4}\log^5_D n) \cdot O(n^{3/2+\eps}) , 
\]
where, as we recall, the prespecified $\eps>0$ can be chosen arbitrarily small, and where
$D=2^{\Theta(1/\eps)}$, with a suitable constant of proportionality.
It then follows that, by slightly increasing~$\eps$, but keeping
it sufficiently small, we can still write the bound as
$O(n^{3/2+\eps})$, with a constant of proportionality of the form $2^{\Theta(1/\eps)}$. 

\begin{remark}
If we care to optimize the resulting bound in terms of $\eps$, we should set 
$\eps=\frac{c}{\sqrt{\log n}}$, for a suitable absolute constant $c>0$, to obtain 
a bound of the form $n^{3/2}\cdot2^{O(\sqrt{\log n})}$. Of course, for this we would 
have to draw curves of degree $D=2^{O(\sqrt{\log n})}$, which are not of constant complexity.  
Our choice of a constant $D$ increases the bound on the number of curves and the complexity
of their arrangement, but ensures that the curves have constant degree. 
\end{remark}

\subsection{Final decomposition into pseudo-trapezoids}
\label{sec:pieces}

Finally, we take the planar map $M_\Delta$, for each triangle $\Delta$, and 
decompose it into regions of constant description complexity, by constructing 
the \emph{trapezoidal decomposition} \cite{4M-book} of $M_\Delta$ in some fixed, 
but arbitrarily chosen ``vertical'' direction within $\Delta$.
Each resulting piece is a ``pseudo-trapezoid,'' with (at most) two vertical sides, 
and ``top'' and ``bottom'' parts, each consisting of a monotone subarc of one 
of the curves we have drawn on~$\Delta$, and thus having degree at most $O(D^2)$. 

The number of trapezoids is proportional to the complexity of $M_\Delta$, 
which in turn is proportional to the number of its vertices and the number 
of points at which the drawn curves have vertical tangents. As the curves have 
degree~$O(D^2)$, each curve can have at most~$O(D^4)$ such tangency points, by Harnack's theorem.

Using the analysis from the preceding subsection, this brings us to the main result of the paper.
\begin{theorem} \label{thm:main}
Let $\T$ be a collection of $n$ pairwise disjoint non-vertical relatively open triangles in $\reals^3$.
Then, for any prescribed $\eps>0$, we can cut the triangles of $\T$ into $O(n^{3/2+\eps})$
pseudo-trapezoids, bounded by algebraic arcs of constant maximum degree $\delta=2^{\Theta(1/\eps)}$,
so that the depth relation among these pseudo-trapezoids is acyclic; here the constant of 
proportionality (and $\delta$) depend on $\eps$.
\end{theorem}
\section{Discussion of recent, related, and future research}
\label{sec:discussion}

In this paper we have essentially settled the long-standing problem of eliminating depth cycles 
in a set of pairwise openly disjoint triangles in $\reals^3$. On the positive side, our solution 
is almost optimal in the worst case, in terms of the number of pieces, as this number is only 
slightly larger than the $\Omega(n^{3/2})$ worst-case lower bound noted in \cite{CEG+}.  
However, a notable disadvantage of our solution is that the cuts are by constant-degree 
algebraic arcs, rather than, ideally, by straight segments. 

One direction for future research is to further tighten the bound, removing
the $\eps$ in the exponent and replacing it by a polylogarithmic factor, 
as in \cite{ArS1dcg} (while keeping the shape of the cut pieces simple).  
In fact, at this point there is no evidence that the correct answer is not simply $\Theta(n^{3/2})$, even for the much simpler case of lines.

As noted above, the BSP partition of \cite{PY} has the stronger property that the depth
relation of the resulting pieces is acyclic with respect to any viewing point or direction.
Our solution does not seem to have this property, so a natural question is whether one 
can cut the triangles into a \emph{subquadratic} number of simple pieces that
have this stronger property, or whether $\Omega(n^2)$ pieces are required, in the worst case,
for this property to hold. 

\paragraph*{A different approach.}
A more recent result of De Berg~\cite{mdb} uses a significantly different approach 
to the problem of cutting triangles into pieces that admit a depth order.  
In his algorithm, the resulting pieces are triangular (the cuts are by straight segments), 
but the bound on the number of pieces produced is only $O(n^{7/4}\polylog n)$, 
significantly higher than what our method produces, albeit with uglier-looking pieces.  
De Berg provides an algorithm whose running time is $O(n^{3.69})$. 

In contrast, at the time when the first version of the paper \cite{triangle-cycles-soda} appeared, our 
algorithm was lacking an effective implementation, because, at that time, there was no available effective
algorithm for constructing the partitioning polynomial of Guth~\cite{Gut}. 

We note that the previous study \cite{ArS1dcg} proposes two other algorithmic approaches 
for computing the cuts in the case of lines: one using the algorithms of 
Har-Peled and Sharir \cite{HPS} or of Solan \cite{So},
and the other using the (slower, but still polynomial, and sharper) approximation algorithm 
of Aronov~et~al.~\cite{ABGM}. 
Unfortunately, neither of these alternative techniques seems so far applicable to the case of triangles.
For the former methods to extend to triangles, one needs an efficient algorithm for testing if a set of, say, triangle fragments has a proper cycle-free depth order, analogous to the algorithm of \cite{dBOS} for segments.
The latter method relies on a close connection between the size of the minimum set of  cuts for breaking cycles among line segments and the size of the minimum set of feedback vertices in a suitable directed graph.  This connection does not appear to have an obvious analogue in the case of triangles.

\paragraph*{Recent progress.}
Two developments took place more recently.
Aronov, Ezra, and Zahl \cite{ArEzZ-soda} have a quadratic-time algorithm for 
effectively constructing a subdivision similar to that promised by 
Proposition~\ref{prop:gut} for an arbitrary set of $n$~bounded-degree curves in $\reals^3$; 
the running time improves to roughly $O(n^{4/3})$ in expectation when applied 
to a set of lines. The subdivision is not purely polynomial, but is a combination 
of a polynomial partitioning and a more traditional cutting (a suitable 
\emph{vertical decomposition}). Because of this feature, the decomposition of~\cite{ArEzZ-soda} 
cannot be used as an immediate drop-in replacement for the one used in this work.

However, even more recently, Agarwal \emph{et al.} \cite{AgArEzZ-poly} constructed a 
general algorithm for effectively computing the partitioning promised by Guth~\cite{Gut} 
in all dimensions and for varieties of any lower dimension; the running time is linear 
in the size of the input. Specializing this algorithm to lines in three dimensions, and 
using it as the drop-in effective and efficient replacement for Proposition~\ref{prop:gut},
yields an $O(n^{3/2+\eps})$-time algorithm for constructing the partitioning in Theorem~\ref{thm:main}. %

\paragraph*{Cutting triangles along curves of degree not bounded by a constant.}
Our bound is slightly larger than that in \cite{ArS1dcg} due to our
choice of a \emph{constant} value, rather than a function of $n$,
for the degree~$D$ of the partitioning polynomials. This is due to our desire to 
partition the triangles into pieces of constant complexity. Choosing for $D$ larger,
non-constant values, such as the value $n^{1/4}$ used in~\cite{ArS1dcg}, would
result in pieces whose shape complexity depends, rather badly, on $n$, not to
mention additional technical problems that arise for such a choice of $D$,
discussed in the following paragraph.

Still, it is interesting to consider the option of optimizing our construction by choosing for~$D$ (and thus also for~$\eps$) a value that depends on $n$
(albeit not as badly as in the preceding paragraph), but then the analysis faces 
a few additional complications. One is that going from the number of curves to 
the number of pieces is not entirely trivial, because the number of endpoints 
and intersections of the curves depends on $D$, so as a result the bound increases 
by some small (but non-constant) factor. Trapezoidation adds another factor that 
depends on the degree of the curves within the triangles. This degree is bounded 
by $O(D^2)$, so the additional factor gained here is also non-constant. 

A more careful examination of the above analysis seems to imply an upper bound on the number 
of pieces produced by our procedure by $n^{3/2}\cdot2^{O(\sqrt{\log n})}$, at the 
cost of using pieces bounded by polynomials of degree $D=2^{O(\sqrt{\log n})}$.
The details needed to make this statement rigorous are left for future work.

\paragraph*{A possible generalization.}
The vertical above/below relation extends naturally to arbitrary pairwise disjoint 
$xy$-monotone algebraic surface patches of constant description complexity, 
with the main difference being that the $\prec$ relation may now contain loops of length 
two, which are impossible if the input consists of pairwise disjoint two-dimensional convex shapes.

Consider a collection of $n$ openly disjoint patches of $xy$-monotone algebraic surfaces, 
each of constant description complexity. Concretely, we assume that each of them is a 
portion of a two-dimensional $xy$-monotone algebraic variety of degree at most $b$, 
bounded by at most $b$ curves of degree at most $b$, for some constant $b>0$.
We believe that our machinery can be applied to eliminate all depth cycles in
such a collection of surface patches, with very few modifications.
The bounded-degree-curve version of Proposition~\ref{prop:gut}, applied to 
the boundaries of the given patches can be used to produce a suitable space partition,
and the rest of the argument, including Lemma~\ref{lemma:new-hope}, should apply in 
this situation more or less verbatim, with various implied constants that depend on $b$.
Verifying and making this claim rigorous is left to future research.

\subsection*{Acknowledgments}
The authors wish to express their gratitude to Saugata Basu for his help in discussing matters of algebra, 
and to Saugata Basu, Sylvain Cappell, Jeff Erickson, and Josh Zahl for some assistance in matters of topology.

\appendix
\section*{Appendix}

In Section~\ref{sec:vertical-decomp}, we recall the definition of a vertical decomposition.
We prove Lemma~\ref{lem:no-disconnect-really} in Section~\ref{sec:where-we-prove-it}.

\section{The vertical decomposition in an arrangement of algebraic surfaces}
\label{sec:vertical-decomp}

In this section, we give, for completeness, a brief description of the notion 
of the \emph{vertical decomposition} of $\reals^3$ with respect to a set $\F$
of $k$ algebraic surfaces, each of degree at most~$D$. 
We follow the definitions in \cite{SA}, and refer the reader also to \cite{collins75} 
for the related notion of a \emph{cylindrical algebraic decomposition}~(\emph{CAD}).
An illustration of the analogous notion of a vertical decomposition in the plane is shown in Figure~\ref{fig:cyl}.
\begin{figure}
  \centering
  \subfigure[]{\includegraphics[page=1]{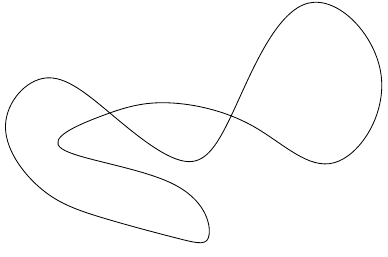}}%
  \qquad%
  \subfigure[]{\includegraphics[page=2]{vv}}%
  \caption{(a) A curve in the plane. (b) The corresponding vertical decomposition, where dots indicate vertices, dashed vertical line segments and rays --- additional vertical edges, and one two-dimensional face is shaded.}
  \label{fig:cyl}
\end{figure}

The \emph{vertical decomposition} $\VV=\VV(\F)$, for a collection $\F$ as above,
is constructed by first drawing a collection $\Gamma$ of curves on the surfaces,
where each curve is either the intersection of two surfaces of $\F$, or is the 
locus of the singular points and the points of $z$-vertical tangency on a single surface.
We then take each curve $\gamma\in\Gamma$ and erect from it a $z$-vertical wall,
which is the union of all the maximal $z$-vertical segments that pass through $\gamma$
and whose relative interiors do not meet any surface of $\F$. That is, we
extend each of the segments up and down until it meets a different point on 
some surface, or else all the way to $z=\pm\infty$.  
The resulting decomposition has the property that each of its 3-cells 
$\nu$ is $xy$-monotone (i.e., it intersects any~$z$-vertical line, if at all, 
in a connected segment), but the $xy$-projection $\nu'$ of $\nu$ need not be 
$x$-monotone, nor even simply connected. In the next step, we form the vertical 
decomposition of the projection $\nu'$ of each first-stage 3-cell $\nu$, 
by drawing a $y$-vertical segment from each singular point and from each 
point of $y$-vertical tangency on the boundary $\bd\nu'$ of~$\nu'$, and 
by extending each such segment up and down, in the $y$-direction, within 
$\nu'$, until it meets another point of $\bd\nu'$, or else all the way to 
$y=\pm\infty$. This yields a decomposition of $\nu'$ into 
\emph{vertical pseudo-trapezoids}. We lift each such pseudo-trapezoid~$\tau'$ in
the $z$-direction (formally, by taking the Cartesian product with~$\reals$) and 
intersect the resulting unbounded prism~$\tau$ with~$\nu$, thereby obtaining a 
\emph{pseudo-prism}~$\tau \cap \nu$. The resulting collection of pseudo-prisms and 
their boundaries constitutes $\VV(\F)$.  For more details, see \cite{SA}.
By construction, every surface in $\F$ is contained in $\VV^{(2)}$.

The decomposition $\VV(\F)$ has the property that each of its (open) 3-cells 
is an open topological 3-ball, which is given by (up to) six inequalities of the 
form $a < x < b$, $f_1(x) < y < g_1(x)$, and $f_2(x,y) < z < g_2(x,y)$, for 
continuous algebraic functions $f_1$, $g_1$, $f_2$, $g_2$, where each of 
$a$, $b$, $f_1$, $g_1$, $f_2$ and $g_2$ may also be $\pm\infty$, in which 
case the corresponding inequality is dropped.

The above description focuses on 3-cells of $\VV=\VV(\F)$ and does not 
specify the exact definition for its 2-, 1-, and 0-faces.  It turns out that, 
for our purposes, such a description is largely unimportant. The 2-skeleton 
$\VV^{(2)}$ of the vertical decomposition, defined as the complement of the 
union of its open 3-cells, is only used in step~(iii) of the curve drawing 
process to create ``wall shadows.'' See the main part of the paper for
the specific $\F$ that we use, and for the complexity analysis of the resulting
decomposition. The curves that we draw are produced by intersecting 
each~$\Delta \in \T$ with the closure of an open cell of~$\VV$, as long as 
the (non-empty) intersection is two-dimensional. 
We then draw the bounding curve of the intersection.

This completes the description of the vertical decomposition.

\section{The proof of Lemma~\ref{lem:no-disconnect-really}}
\label{sec:where-we-prove-it}

In this section we utilize standard notions from topology and algebraic topology,
as documented, e.g., in Spanier \cite{Spanier}. We will need the following elementary fact.
\begin{fact}
  \label{fact:open-is-open}
  Any open set $U \subset \reals^d$ is a
  disjoint union of its path-connected components; each such component is an open set in~$U$.
\end{fact}

\begin{proof}
  Define an equivalence relation $\sim$ on points of $U$, so that $x \sim y$, 
  if there is a continuous path in $U$ from $x$ to $y$. The equivalence classes 
  are called the \emph{path-connected components} (or \emph{path components}, for short) of~$U$. 
  This defines a partition of $U$ into disjoint non-empty subsets, $U =\bigsqcup_j \ U_j$, 
  where each $U_j$ is path connected, while there is no continuous path in~$U$ between 
  points of $U_j$ and~$U_k$, for any~$j\neq k$.

  Fix a path component $U_j$ and let $x \in U_j$. The open set $U$ is a union of open balls
  so there is an open ball, say $B_r(p)$, with center $p$ and radius $r>0$ in~$U$, containing $x$. Since
  $B_r(p)$ is convex and contains $x$, the line segment between $x$ and any point in $B_r(p)$
  is a continuous path in $U$. This proves that $B_r(p) \subset U_j$, so necessarily $U_j$ 
  is a union of open balls of $\reals^d$, and therefore $U_j$  is an open set of $\reals^d$. 
\end{proof}

\paragraph*{Lemma setup.}

We introduce a few constructs and definitions. Fix a positive real number $M >0$. 
Let $a,b \in \reals$ with
$-M < a < b < +M$ and let $f_1,g_1 \colon (a,b) \to \reals$ be continuous functions
satisfying
\[
  -M < f_1(x) < g_1(x) < +M,
\]
for all $x \in (a,b)$.  Define the open bounded set $P \subset \reals^2$ as
\[
  P \coloneqq \{ (x,y) \mid x \in (a,b), f_1(x) < y < g_1(x) \}.
\]
Let $f_2,g_2 \colon P \to \reals$ be continuous functions satisfying
\[
  -M < f_2(x,y) < g_2(x,y) < +M,
\]
for all $(x,y) \in P$.
Now let $\nu$ be the open set in $\reals^3$ defined by
\[
  \nu \coloneqq \{ (x,y,z) \mid  (x,y) \in P, f_2(x,y) < z< g_2(x,y) \}.
\]
Each (bounded open) 3-cell of the vertical decomposition $\VV=\VV(Z(f))$ of $Z(f)$, 
as discussed in Appendix~\ref{sec:vertical-decomp}, is of the above form, for a suitable choice of $M$.

Finally, let $L \colon \reals^2 \to \reals$ be any continuous 
function, and let the non-vertical surface~$\Pi=\Pi(L)$ be the graph of~$L$ over~$\reals^2$, namely
\[
  \Pi \coloneqq  \{(x,y,L(x,y)) \mid (x,y) \in \reals^2 \}.
\]
Assume that $\Pi$ intersects $\nu$.

Note that $(x,y) \rightarrow  (x,y,L(x,y))$ defines a homeomorphism of~$\reals^2$ to $\Pi$. 
In particular, a set is open in $\reals^2$ if and only if its image is open in $\Pi$, 
with the subspace topology inherited from~$\reals^3$, namely the topology whose open sets 
are exactly the intersections of open sets of $\reals^3$ with $\Pi$.
Hence, the set $\nu \cap \Pi$, which is open in the subspace topology of $\Pi$,
is homeomorphic to an open set of~$\reals^2$. Appealing to Fact~\ref{fact:open-is-open},
this open set in $\reals^2$ is a disjoint union of open and path-connected spaces 
(its path components). Consequently, the open set $\nu \cap \Pi$ in $\Pi$ is a 
disjoint union of its open and path-connected components. 
Let $\slice$ be one of these components of $\nu \cap \Pi$. 

\begin{lemma}
  \label{lemma:new-hope}
  $\nu \setminus \slice$  is an open set in $\reals^3$ with exactly two  open path-connected components.
  In particular, any continuous path in $\nu$ connecting points in different components of $\nu \setminus \slice$ must meet $\slice$.
\end{lemma}
\begin{proof}
  Firstly, note that $\Pi$ is closed in $\reals^3$, while $\nu$ is open in $\reals^3$, so $\nu \setminus \Pi$ is open in $\reals^3$.

  Now consider the decomposition of the open set $\nu \cap \Pi$ in $\Pi$ into its chosen 
  path-component~$V$ and the union~$W$ of its remaining path-components; $W$ may be empty.
  By Fact~\ref{fact:open-is-open},
  both~$V$~and~$W$, as  unions of path-components of the open set $\nu \cap \Pi$, are open in the subspace
  topology of $\Pi$. Hence, there is an open set $W^\#$ of $\reals^3$ with $W = W^\# \cap \Pi$.
  In particular, $W^\# \cap \nu$ is an open (in $\reals^3$) subset of $\nu \setminus V$, and the sequence of equalities
  \[
    \nu \setminus V = (\nu \setminus V) \cup ( W^\# \cap \nu)= (\nu \setminus \Pi) \cup ( W^\# \cap \nu)
  \]
  shows that $\nu \setminus V$ is an open set in $\reals^3$, as claimed.

  It remains to show that $\nu \setminus V$ consists of exactly two path components.  
  Then, by definition of path components, any  continuous path in $\nu $ between
  points in different components of~$\nu \setminus V$ must necessarily cross~$V$.

  The set $\nu$ is homeomorphic to the open unit cube $(0,1) \times (0,1) \times (0,1)$ 
  by the homeomorphism $\Psi \colon \reals^3 \to \reals^3$ defined by
  \[
    \Psi(x,y,z) = \left ( \frac{x-a}{b-a}, \ \frac{y-f_1(x)}{g_1(x)-f_1(x)}, \  \frac{z-f_2(x,y)}{g_2(x,y)- f_2(x,y)} \right ).
  \]
  The projection $\rho \colon \reals^3 \rightarrow \reals^2$ defined by $\rho(x,y,z) = (x,y)$ maps
  $\Pi$ homeomorphically to~$\reals^2$. 
  Also, since $V$ is open, path-connected, and connected, so is
  $\hat{V} \coloneqq \rho(V) \subseteq \reals^2$.

  Denote by $T$ the intersection
  \[
  T \coloneqq \rho^{-1}(\hat{V}) \cap \nu = \{ (x,y,z) \mid ( x,y) \in \hat{V}  \text{ and } f_2(x,y) < z < g_2(x,y) \},
  \]
  which is an open set in $\reals^3$ as $\rho^{-1}(\hat{V}) = \hat{V} \times \reals$ and  $\nu$ are open in $\reals^3$.
  Note that for $(x,y) \in \hat{V}$, $(x,y,L(x,y)) $ lies in $V$ and hence in $\nu$. 
  Thus, the inequalities $f_2(x,y) < L(x,y) < g_2(x,y)$ hold for $(x,y) \in \hat{V}$.

  In particular, $T$ is the disjoint union $T  = V \sqcup X \sqcup Y$ with
  \begin{align*}
    X \coloneqq {}&\{(x,y,z) \mid f_2(x,y) < z < L(x,y),\ 
                  (x,y ) \in \hat{V}  \}, \text{ and} \\
    Y \coloneqq {}& \{(x,y,z) \mid  L(x,y) < z < g_2(x,y),\ %
                  (x,y ) \in \hat{V}\}.
  \end{align*}
  $X$ and $Y$ are open sets in $\reals^3$. The spaces $X$, $Y$, and $T$ are path connected, 
  by the following argument: Define $V_1 \subset X, V_2 \subset Y$ by
  \begin{align*}
    V_1 \coloneqq {} &\{(x,y,(L(x,y)+f_2(x,y))/2) \mid (x,y) \in \hat{V} \}, \text{ and}\\
    V_2 \coloneqq {} &\{(x,y,(L(x,y)+ g_2(x,y))/2) \mid (x,y) \in \hat{V} \}.
  \end{align*}
As is easily seen, under the projection $\rho$, 
the sets $V_1$ and $V_2$ each map homeomorphically onto~$\hat{V}$.
Since $\hat{V}$ is open and path connected, so is each of $V_1$ and $V_2$, in the subspace topology.

Consider $X$. For points $(x,y,z), (x',y',z')$ in $X$,
the vertical line segment from $(x,y,z)$ to $R \coloneqq (x,y,(L(x,y)+f_2(x,y))/2)$
is fully contained in $X$ and provides a path from $(x,y,z)$ to the point $R$
of $V_1$. Since $V_1$ is path connected, there is a continuous path from $R$ to 
$S \coloneqq (x',y',(L(x',y')+f_2(x',y'))/2)$ in $V_1$. 
The vertical line segment from $(x',y',z')$ to $S$ provides a continuous  
path in $X$ between these two points. Concatenating these three paths provides
a continuous path in $X$ from $(x,y,z)$ to $(x',y',z')$. This proves that $X$ is path connected.

To prove the path-connectedness of $Y$ (respectively, $T$), we apply the same method, 
replacing $V_1$ by $V_2 \subset Y$ (respectively, by $V$).

Now note that $T\setminus V = X \sqcup Y$ with both sets $X$ and $Y$ 
non-empty and open. Hence, $T$ is not connected, and necessarily not path connected. 
In particular, there is no continuous path in $X \sqcup Y$ from a point of $X$
to a point of $Y$.  This proves that $T \setminus V = X \sqcup Y$ has exactly two path components.

\bigskip

We now turn to homological techniques, using singular homology with integer coefficients.
We assume some basic familiarity of the reader with these techniques, which can be found,
in full detail, e.g., in Spanier \cite{Spanier}.

From the definition of the reduced zero-dimensional singular homology $\tilde{H}_0(X;\ints)$ 
of a topological space~$X$ with integer coefficients, it easily follows that 
$\tilde{H}_0(M;\ints) = \{0\}$ if and only if $M$ is path connected. Additionally, 
$\tilde{H}_0(M;\ints) = \ints$ if and only if $M$ is the disjoint union of exactly 
two path components. In the latter case, the generator of $\tilde{H}_0(M;\ints) = \ints$ 
is the reduced zero chain $1 \cdot p - 1 \cdot q$ with $p$ a point in one path component
and $q$ a point in the other.

In particular, since  $V$, $V_1$, $V_2$, $X$, $Y$, $T$, and $\nu$ are all path 
connected, their reduced singular zero-dimensional homologies with integer coefficients,
$\tilde{H}_0(\cdot; \ints)$, vanish, i.e.,
\[
  \tilde{H}_0(V;\ints) = \tilde{H}_0(V_1;\ints) = \tilde{H}_0(V_2;\ints) = 
  \tilde{H}_0(X;\ints) = \tilde{H}_0(Y;\ints) =\tilde{H}_0(T;\ints)
  =\tilde{H}_0(\nu;\ints) =\{0\}.
\]

We have proven that $T \setminus V$ consists of two path-connected components, 
$X$ and $Y$; in particular, no path in $T$ connects $X$ to $Y$ without visiting $V$.
Equivalently stated, $\tilde{H}_0(X \sqcup Y;\ints) = \ints$, with generator 
$1 \cdot p - 1 \cdot q$ for any $p \in X$ and $q \in Y$. Fix such a choice of points $p$ and $q$.

\bigskip

When a topological space $X_1\cup X_2$ can be written as a union of two open subsets~$X_1$ and~$X_2$ (in $X_1 \cup X_2$), the following \emph{Mayer-Vietoris sequence} is \emph{exact}:
  \begin{equation}
    \label{eq:m-v}
    \tilde{H}_1(X_1 \cup X_2;\ints) \to \tilde{H}_0(X_1 \cap X_2;\ints)
    \to \tilde{H}_0(X_1;\ints) \oplus \tilde{H}_0(X_2;\ints) \to \tilde{H}_0(X_1 \cup X_2;\ints).
  \end{equation}
  See, for example, the more general Theorem~4.6.3 in Spanier \cite[page~188]{Spanier}, 
  for the precise descriptions of the maps in the sequence. We take the open sets $X_1,X_2$ 
  to be $X_1 = T $\ and  $ X_2 = \nu \setminus V$.
  The union $X_1 \cup X_2$ is the open set $\nu$ in $\reals^3$, $X_1,X_2$ are open in $\reals^3$, 
  as observed above, and thus also in $\nu = X_1\cup X_2$, and 
  $X_1 \cap X_2 = T \setminus V = X \sqcup Y$ is also open in $\reals^3$, and so in $\nu$.  
  Therefore, the exact sequence~\eqref{eq:m-v} reads in our case
  \begin{gather*}
    \tilde{H}_1(\nu;\ints) \to \tilde{H}_0(X \sqcup Y;\ints)
    \to \tilde{H}_0(T;\ints) \oplus \tilde{H}_0(\nu \setminus V;\ints) \to \tilde{H}_0(\nu;\ints), \text{\ \ or}\\[0.5ex]
    \{0\} \rightarrow \ints  \rightarrow \{0\} \oplus \tilde{H}_0(\nu \setminus V;\ints)  \rightarrow \{0\}.
  \end{gather*}
  Above $\tilde{H}_1(\cdot;\ints)$ denotes the reduced one-dimensional singular homology 
  with integer coefficients, which, informally, represents non-trivial one-dimensional 
  cycles in~$\nu$. Here we have $\tilde{H}_1(\nu;\ints)=\{0\}$ because $\nu$ is contractible.
  The exactness of the sequence implies the isomorphisms 
  $$
  \ints =  \tilde{H}_0(X \sqcup Y;\ints) \cong \{0\} \oplus H_0(\nu \setminus V;\ints) \cong H_0(\nu \setminus V;\ints) ,
  $$
  where the generator of the group on the left-hand side is precisely the reduced zero chain 
  $1 \cdot p - 1 \cdot q$ mentioned above.
  Therefore $\nu \setminus V$ has exactly two path-connected components $X^\# \supset X$ 
  and $Y^\# \supset Y$, and every path in $\nu$ connecting points in different components 
  must therefore meet $V$.

  By Fact \ref{fact:open-is-open} applied to the open set $\nu \setminus V$ in $\reals^3$,
  the two path-components $X^\#, Y^\#$ are necessarily open. 

  This completes the proof of Lemma~\ref{lemma:new-hope}.
\end{proof}

Finally, we are ready to present the proof of Lemma~\ref{lem:no-disconnect-really}.
\begin{proof}[Proof of Lemma~\ref{lem:no-disconnect-really}]
Let $\nu$ be a bounded 3-cell of the vertical decomposition $\VV=\VV(Z(f))$,
as defined in Appendix~\ref{sec:vertical-decomp}. It is sufficient to consider 
only bounded cells: Indeed, our input triangles are bounded, so we can place 
them inside a sufficiently large box, form the vertical decomposition within 
the box, e.g., by adding the planes supporting the box faces to $Z(f)$, and 
disregard all cells of~$\reals^3 \setminus Z(f)$ lying outside the box, without 
affecting the rest of the argument.

Let $\Delta \in \T$ be one of the triangles that slice $\nu$, and let $h_\Delta$ 
be its supporting plane. For a contradiction, suppose that there exists a loop 
$\pi$ of the form~\eqref{eq:pi} that is fully contained in $\nu$ and visits $\Delta$.  
(Note that, by construction, $\pi$~meets~$\Delta$ a~finite number of times, 
in the sense that $\pi^{-1}(\Delta)$~consists of a~finite number of connected components.)
Let $p$ be a point on $\pi$ in~$\Delta \cap\nu$. Define $L$ to be the linear function 
whose graph $\Pi$ (in the terminology of the proof of Lemma~\ref{lemma:new-hope}) 
is~$h_\Delta$. Let $\slice$ be the path component of~$h_\Delta \cap \nu$ containing $p$; 
since $\Delta$ slices $\nu$, this component coincides with the path component 
of~$\Delta \cap \nu$ containing~$p$.

Lemma~\ref{lemma:new-hope} shows that $\nu \setminus V$ consists of two path components, 
the ``lower''~$X^\#$ and the ``upper''~$Y^\#$, and any path in~$\nu$ passing from one 
component to the other has to meet~$\slice$ between them.  In particular, since $\pi$ 
is a closed loop, it must travel from $X^\#$ to~$Y^\#$ (necessarily via~$V$) as many 
times as it travels from~$Y^\#$ back to~$X^\#$ (again, necessarily via~$V$). 
However, by construction, $\pi$ can only pass from below $\slice \subset \Delta$ 
to above it (which may be witnessed by a subpath $\pi_i$ contained in $\Delta_i=\Delta$  
and the corresponding nearby portions of $\pi$ just before and after $\pi_i$, or 
when one of the upward jumps intersects $\slice$), and not in the opposite direction.

This yields a contradiction unless $\pi$ avoids $V$ entirely, thereby concluding the proof of the lemma.
\end{proof}

\end{document}